\documentclass[a4,10 pt]{amsart}

\usepackage{fullpage}
\usepackage{cancel}
\usepackage[pdftex]{graphicx}

\usepackage{bm,amsmath}
\usepackage{mathrsfs}
\usepackage{amsfonts}
\usepackage{color}
\usepackage{amssymb}
\usepackage{amsthm}
\usepackage{latexsym}
\usepackage{fancyhdr}
\usepackage{subfigure}
\usepackage{stmaryrd}

\newtheorem{thrm}{Theorem}[section]
\newtheorem*{thrm*}{Theorem}
\newtheorem{lemm}[thrm]{Lemma}
\newtheorem{coro}[thrm]{Corollary}
\newtheorem{prop}[thrm]{Proposition}

\newtheorem{defi}[thrm]{Definition}

\newtheorem{nota}[thrm]{Notation}

\newtheorem{rmrk}[thrm]{Remark}

\newcommand{\Er}{\mathbb{R}}

\begin{document}

\title[Stability]{Spectra of Laplacian matrices of weighted graphs:\\ structural genericity properties}
\author{Camille Poignard, Tiago Pereira, Jan Philipp Pade} 
\address{ICMC S\~ao Carlos, University of S\~ao Paulo (SP), Brasil}
\email{camille.poignard@gmail.com}
\date{March the 18{th}, 2017}

\begin{abstract}
This article deals with the spectra of Laplacians of weighted graphs. In this context, two objects are of fundamental importance for the dynamics of complex networks: the second eigenvalue of such a spectrum (called algebraic connectivity) and its associated eigenvector, the so-called Fiedler vector. Here we prove that, given a Laplacian matrix, it is possible to perturb the weights of the existing edges in the underlying graph in order to obtain simple eigenvalues and a Fiedler vector composed of only non-zero entries.
These \textit{structural} genericity properties with the constraint of not adding edges in the underlying graph are stronger than the classical ones, for which arbitrary structural perturbations are allowed.
These results open the opportunity to understand the impact of structural changes on the dynamics of complex systems. 
\end{abstract}

\noindent

\maketitle



\section{Introduction}
Many dynamical phenomena observed in real world networks take place on weighted graphs \cite{barrat2008, Arenas2008}. 
Synchronization in networks of lasers is an important example as it increases laser stability \cite{flunkert2010}. Another example is given by stability in power grids, where certain coupling configurations can lead to the desynchronization of generators and thereby to major blackouts \cite{motter2013}. Recently, random walks on graphs have attracted much attention \cite{Lovasz1993, Chung97, Estrada2015} which has been of interest for graph community detection (see \cite{delvenne2010} and \cite{Estrada2011a} as well).
Further important examples can be found in neuroscience \cite{Sporns2001, Sporns2009} and in epidemiology \cite{Wang2007}. 
These dynamical  processes are intimately related to the spectrum of the Laplacian matrix associated with the graph. 

Simplicity {of} eigenvalues plays a major role as it guarantees good properties of the underlying dynamics such as exponentially and uniformly fast convergence towards synchronization in diffusively coupled networks \cite{Pereira2014} and convergence to the stationary measure in random walks \cite{Aldous2002}.  In fact, when the eigenvalues fail to be simple the underlying dynamics can be pathological \cite{nishikawa2010}.
The eigenvectors of the Laplacian matrix
also have a fundamental impact on the system. In particular, the so-called Fiedler vector associated with the algebraic connectivity (i.e the second eigenvalue) of this matrix is of importance for spectral graph partitioning \cite{Merris1998,Von2007tutorial} and for synchronization \cite{Pecora,Pereira2014}. Indeed, when the Fiedler vector has non-zero entries one can design structural changes that make synchronization unstable \cite{PerPade2015,PadeHart2015}.

The purpose of this article is to deal with generic properties of graph Laplacians spectra.
Since the $70$'s and the seminal work of Fiedler, graph spectra have attracted a great deal of attention \cite{Fiedler1973,cvetkovic1980,mohar1997,powers1988,biyikougu2007,Estrada2005, Deville2014}. 
Here we focus on generic properties under the constraint of keeping the graph structure unchanged and only slightly modifying the positive weights. 

In general, given any Laplacian matrix (in fact any matrix) it is always possible to perturb its entries to obtain simple eigenvalues. Typically, such perturbations will change the underlying graph drastically. In fact, these perturbations usually cascade to the whole graph resulting in a massively connected graph or even a fully connected one.
However, in the context of many dynamical networks it is important to keep the initial graph topology. For instance, adding new connections to power grids is very costly and can lead to instabilities \cite{witthaut2012}. In general, the links have assigned physical meaning and while we can slightly change the weights, we cannot introduce new links. In other words, we look for generic properties when we don't add new edges to the initial graph and only slightly change the weights on already existing edges.\\  
Hence, given a connected weighted and undirected graph our main goal is to address the two questions:
\begin{center}
\begin{minipage}[c]{14.5cm}
\textbf{(Q1)} \textit{Can we perturb the weights of the existing edges to obtain a Laplacian matrix with simple eigenvalues?} 

\textbf{(Q2)} \textit{Can we perturb the weights of the existing edges to obtain a Laplacian matrix with Fiedler vector having non-zero entries?}\\ 
\end{minipage}
\end{center}

The manuscript gives positive answer to these two questions. In fact, we prove stronger results. 
The set of graph Laplacians such that these two questions have positive answer has full Lebesgue measure. We will show that for directed graphs the answer to \textbf{(Q1)} is also positive. 

The manuscript is organized as follows. In Section \ref{notations}, we introduce our notations and the notion of structural genericity for sets of symmetric zero-row sum real matrices. This notion permits us to formulate a density result with constraint for undirected graphs. In Section \ref{mainResults} we state our main results (Theorem \ref{one} and Theorem \ref{two}) providing positive answer to \textbf{(Q1)} and \textbf{(Q2)} by showing that having a simple spectrum and eigenvectors with non-zero entries are structurally generic properties.
Then, in Section \ref{proof1} we prove Theorem \ref{one} and in Section \ref{proof2} we prove Theorem \ref{two}. Finally, in Section \ref{Digraphs} we generalize our first result on simplicity of eigenvalues to directed graphs (Theorem \ref{one'}).

\section{Notations and Definitions}\label{notations}

\subsection{Weighted and Undirected Graphs}
Let us recall some basic notions of algebraic graph theory (we refer the reader to \cite{biggs1993,cvetkovic1980} and the references therein for a detailed introduction on the subject).\\
A \textit{simple graph} with $n$ vertices or nodes is a pair $\left(\mathcal{V},{E}\right)$ formed by a vertex set $\mathcal{V}=\left(v_1,\cdots,v_n\right)$ and an edge set $E \subset \{1,\cdots, n\}\times \{1,\cdots, n\}$, where a pair $\left(i,j\right)$ belonging to $E$ is called an \textit{edge} or \textit{link} between vertex $v_i$ and vertex $v_j$ (or in a short way, $\left(i,j\right)$ is an edge between node $i$ and node $j$). In particular, there is no \textit{loop} linking a node to itself. A graph is \textit{weighted} if any of its edges $\left(i,j\right)$ is associated to a number $w_{i,j}>0$ (called \textit{weight}).
When $\left(i,j\right)$ is not in $E$, we set $w_{i,j}=0$. 
A graph is said to be \textit{undirected} when for any pair $\left(i,j\right)$ belonging to its edge set $E$ we have that $\left(j,i\right)$ belongs to $E$ as well and $w_{i,j}=w_{j,i}$.  \\

\noindent
{\it Connected Graphs}. An undirected graph $\mathcal{G}=\left(\mathcal{V},{E}\right)$ is connected, if for any two nodes $i$ and $j$ of $\mathcal{G}$, there exists a \textit{path} $\{i=i_1,\cdots,i_p=j\}$ of connected nodes (i.e nodes successively connected by some edges of $E$) between node $i$ and node $j$.

A \textit{tree} is an undirected graph for which any two nodes $i$ and $j$ are connected by exactly one path.
Geometrically, it means there is no cycle in such a graph: for this reason trees are also referred to as acyclic graphs. 
A \textit{spanning tree} is a tree of which set of nodes equals $\mathcal{V}$, and for which the set of edges is included in $E$. A \textit{rooted spanning tree} is a spanning tree for which one node has been designated as a root.

\subsection{Notations on matrices}
We will denote by $\mathcal{M}_n \left(\Er\right)$ the set of square real matrices of size $n$, by $\mathcal{Z}_n \left(\Er \right)$ the vector subspace of $\mathcal{M}_n \left(\Er\right)$ formed by zero-row sum real matrices of size $n$, and by $\mathcal{S}_n \left(\Er \right)$ the vector subspace of $\mathcal{M}_n \left(\Er \right)$ formed by symmetric real matrices of size $n$. Given a subset $\mathcal{M}$ of $\mathcal{M}_n\left(\Er\right)$, ${\mathcal{M}}^c$ will be its complement set. The notation $\lvert \lvert \cdot \rvert \rvert$ will stand for a norm over the set of matrices of $\mathcal{M}_n\left(\Er\right)$. Any norm can be used in our study.

Given a matrix $M$ in $\mathcal{M}_n\left(\Er\right)$, the notation $det \left(M\right)$ will stand for the determinant of $M$, $\chi_{M}$ will stand for the characteristic polynomial of the matrix $M$, ${}^{\textbf{t}}M$ will stand for the transpose matrix of $M$, and $\textbf{Com}\left(M\right)$ for the comatrix of M, formed by the cofactors of these matrices. In other words the element $c_{i,j}$ of $\textbf{Com}\left(M\right)$ is given by:
\begin{align*}
c_{i,j}=\left(-1\right)^{i+j} det\left( M\left(i\vert j\right)\right),
\end{align*}
where $M\left(i\vert j\right)$ is the submatrix of $M$ obtained by deleting the $i$-th row and the $j$-th column.

$\Er_+^n$ will denote the subset of vectors for which the coordinates are positive or null, and ${\Er_+^*}^n$ the one for which the coordinates are all strictly positive. Given a vector $X$ in $\Er^n$, its $i$-th entry will be denoted by $X_{\left(i\right)}$. Given a subspace $A$ of $\Er^n$ we will denote by ${\mathscr{L}}_{\vert A}$ the Lebesgue measure restricted to $A$.

\subsection{The set of Laplacian matrices of weighted graphs}
Given an undirected weighted graph $\mathcal{G}$ its Laplacian matrix $L$ is the square matrix of size $n$ defined by  
$$L=D-W,$$ where $W$ is the weighted \textit{adjacency matrix} representing the weights of the edge set $E$ i.e $W = (w_{i,j})_{i,j=1}^n$,
and where
$D = $diag$(D_1,\cdots,D_n)$ with $D_i=\sum^n_{j=1}w_{i,j}$, 
is the matrix of degrees of $\mathcal{G}$. The Laplacian matrix $L$ is a positive semi-definite operator:
\begin{prop}\label{prope}
Let $L$ a Laplacian matrix of an undirected graph with $n$ nodes. Then, the spectrum $\mathfrak{S}\left(L\right)$ has the form:
\begin{align*}
\mathfrak{S}\left(L\right)=\{0 = \lambda_1 \leq \lambda_2\leq \cdots \leq \lambda_n\}.
\end{align*}
Moreover, the multiplicity of the eigenvalue $0$ equals the number of connected components of the graph.
\end{prop}
The first part of this result follows from the Gerschgorin's disk theorem. The multiplicity of the zero eigenvalue is readily obtained by noticing that constant vectors are eigenvectors associated with this eigenvalue. The eigenvalue $\lambda_2$ (possibly equal to zero) is called  \textit{algebraic connectivity} of the graph. It is closely related to constants which are important for characterizing the topology of a graph, such as the diameter or the isoperimetric number.

\begin{nota}
We will denote by $\mathcal{W}$ the set of Laplacian matrices associated to connected undirected weighted graphs with $n$ nodes.
We will denote by:
\begin{itemize}
\item[]$\mathcal{W}_s$ the subset of $\mathcal{W}$ formed by  Laplacian matrices having only simple eigenvalues,
\item[]$ \mathcal{W}_0$ the subset of $\mathcal{W}$ formed by Laplacian matrices whose Fiedler vector has at least one zero entry.
\end{itemize}
\end{nota}

In Theorem \ref{one} we prove a \textit{structural density} property. Before stating this result, we need to introduce the following notion of perturbations of symmetric zero-row sum matrices:

\begin{defi}[Structural Perturbations]\label{defun}
For any matrix $M$ in $\mathcal{S}_n\left(\Er \right)\cap \mathcal{Z}_n\left(\Er \right)$ and any tuple $\mathscr{E}=\left(\epsilon_{i,j}\right)_{1\leq i<j\leq n}$ in $\Er^{\frac{n \left(n-1\right)}{2}}$ we define the matrix $\underline{M}\left(\mathscr{E}\right)$ in $\mathcal{S}_n\left(\Er \right)\cap \mathcal{Z}_n\left(\Er \right)$ by:
\begin{align*}
\forall 1\leq i< j\leq n, \,\,&\underline{M}\left(\mathscr{E}\right)_{i,j}=\begin{cases}
-\epsilon_{i,j} \text{ if } M_{i,j} \neq 0\\
0 \text{ else}\\
\end{cases}\\
\,\,&\underline{M}\left(\mathscr{E}\right)_{j,i}=\underline{M}\left(\mathscr{E}\right)_{i,j}\\
\,\,&\underline{M}\left(\mathscr{E}\right)_{i,i}=-\sum^n_{\substack{j=1\\
j \neq i}}\underline{M}\left(\mathscr{E}\right)_{i,j}.
\end{align*}
\end{defi}

For example, consider the following Laplacian matrix in $\mathcal{S}_{4}\left( \Er\right)$:
\begin{align*}
L=\begin{bmatrix}
3/2&-1&0&-1/2\\
-1&2&-1&0\\
0&-1&4/3&-1/3\\
-1/2 &0 &-1/3& 5/6\\
\end{bmatrix}.
\end{align*}
Then given any tuple $\mathscr{E}=\left(\epsilon_{i,j}\right)_{1\leq i<j\leq 6}$ in $\Er^{6}$, the matrix $\underline{L}\left(\mathscr{E}\right)$ is equal to:
\begin{align*}
\underline{L}\left(\mathscr{E}\right)=\begin{bmatrix}
\epsilon_{1,2}+\epsilon_{1,4}&-\epsilon_{1,2}&0&-\epsilon_{1,4}\\
-\epsilon_{1,2}&\epsilon_{1,2}+\epsilon_{2,3}&-\epsilon_{2,3}&0\\
0&-\epsilon_{2,3}&\epsilon_{2,3}+\epsilon_{3,4}&-\epsilon_{3,4}\\
-\epsilon_{1,4} &0 &-\epsilon_{3,4}& \epsilon_{1,4}+\epsilon_{3,4}\\
\end{bmatrix}.
\end{align*}

\begin{rmrk}
Given a Laplacian matrix $L$, the matrix $\underline{L}\left(\mathscr{E}\right)$ may not be Laplacian, for its off-diagonal entries may not be negative. 
However, $\underline{L}\left(\mathscr{E}\right)$ is a Laplacian matrix for $\mathscr{E}$ in $\Er_{+}^{{\frac{n \left(n-1\right)}{2}}}$.
We emphasize the fact that, even for a tuple $\mathscr{E}$ in ${\Er_+^*}^{\frac{n \left(n-1\right)}{2}}$, the graph associated to the Laplacian matrix $\underline{L}\left(\mathscr{E}\right)$ is a subgraph of $\mathcal{G}$ (or $\mathcal{G}$ itself). It does not have more edges than the graph $\mathcal{G}$.
\end{rmrk}

To formulate Questions \textbf{(Q1)} and \textbf{(Q2)} we need a new notion of density with constraint for sets of real matrices, namely the constraint of keeping the graph structure of a given matrix.

\begin{defi}[Structural density]
Let $\mathcal{M}_0\subset \mathcal{M}$ be two subsets of $\mathcal{S}_{n}\left( \Er\right)\cap \mathcal{Z}_n\left(\Er \right) $. We say that $\mathcal{M}_0$ is \textit{structurally dense} in $\mathcal{M}$ if the following holds:
\[
\forall M \in \mathcal{M},\,\, \forall \epsilon_0 >0,\,\, \exists \mathscr{E} \in \Er^{{\frac{n \left(n-1\right)}{2}}},\,\,\text{such that } M+\underline{M}\left( \mathscr{E}\right) \in\mathcal{M}_0 \text{ and } \lvert \lvert \underline{M}\left( \mathscr{E}\right)  \rvert \rvert< \epsilon_0.
\]
\end{defi}
Obviously, one could extend the definition of structural density to the entire set $\mathcal{S}_{n}\left( \Er\right)$: to do this, it suffices to delete the condition $$\underline{M}\left(\mathscr{E}\right)_{i,i}=-\sum^n_{\substack{j=1\\
j \neq i}}\underline{M}\left(\mathscr{E}\right)_{i,j}$$ in Definition \ref{defun} above. We don't consider this extended definition here as our article focuses exclusively on Laplacian matrices, for which the structural perturbations correspond precisely to Definition \ref{defun}. Concerning the non symmetric case, the notion of structural density for matrices in the complement set $\mathcal{S}_{n}\left( \Er\right)^c$ of real non symmetric matrices is dealt with in Section \ref{Digraphs}, where we consider non symmetric graph Laplacians.
To this \textit{structural density} notion, corresponds a new notion of \textit{structural genericity} property:

\begin{defi}[Structural genericity]
A property $\mathscr{R}_0$ for matrices belonging to a subset $\mathcal{M}$ of $\mathcal{S}_{n}\left( \Er\right)\cap \mathcal{Z}_{n}\left( \Er\right)$ is said to be structurally generic in $\mathcal{M}$ if:\\
 - The set of matrices in $\mathcal{M}$ satisfying $\mathscr{R}_0$ is structurally dense in $\mathcal{M}$.\\
 - Given a matrix $M\in \mathcal{M}$, the set of tuples $\mathscr{E}$ in $\Er^{{\frac{n \left(n-1\right)}{2}}}$ such that $\underline{M}\left( \mathscr{E}\right)$ belongs to $\mathcal{M}$ and does not satisfy $\mathscr{R}_0$ is of Lebesgue measure zero in $\Er^{{\frac{n \left(n-1\right)}{2}}}$.
\end{defi}

\section{Main Results}\label{mainResults}
We can now enunciate the first result of this paper, which gives a positive answer to \textbf{(Q1)}:
\begin{thrm}\label{one}
The property of having only simple eigenvalues is structurally generic in the set $\mathcal{W}$, that is:\\
 - The set $\mathcal{W}_s$ is structurally dense in $\mathcal{W}$.\\
 - For any Laplacian matrix $L$ in $\mathcal{W}$, the set of tuples $\mathscr{E}$ in $\Er^{{\frac{n \left(n-1\right)}{2}}}$ such that 
$\underline{L}\left( \mathscr{E}\right)$ belongs to ${\mathcal{W}_s}^c \cap \mathcal{W}$ is of Lebesgue measure zero in $\Er^{{\frac{n \left(n-1\right)}{2}}}$.
\end{thrm}

Hence, changing only the existing weights of a graph Laplacian in $\mathcal{W}$ leads to simple spectrum. Actually, in Section \ref{Digraphs} we will prove a similar result in the case of directed graphs (see Theorem \ref{one'}). Our next results provides an affirmative answer to \textbf{(Q2)}, which asks how big the complement set ${\mathcal{W}_0}^c\cap \mathcal{W}$ is in $\mathcal{W}$:
\begin{thrm}\label{two}
The property of having a Fiedler vector with  only non-zero components is structurally generic in $\mathcal{W}$:\\
 - The set ${\mathcal{W}_0}^c\cap \mathcal{W}$ is structurally dense in the set $\mathcal{W}$.\\
 - For any Laplacian matrix $L$ in $\mathcal{W}$, the set of tuples $\mathscr{E} \in \Er^{{\frac{n \left(n-1\right)}{2}}}$ such that 
$\underline{L}\left( \mathscr{E}\right)$ belongs to $\mathcal{W}_0$  is of Lebesgue measure zero.
\end{thrm}

Actually, we prove  a more general extension of Theorem \ref{two}: the property of having a basis of eigenvectors having only non-zero components is structurally generic in $\mathcal{W}$. This is stated in Corollary \ref{cor1}. 
\subsection{Sketches of the proofs}
To prove Theorem \ref{one} we follow two main steps. First, we construct an iterative process to obtain a Laplacian matrix with simple eigenvalues $\underline{L}\left(\mathscr{E}\right)$ over a graph with the same structure as the original graph. We start from the longest path inside one of its spanning trees, and iteratively include new branches while controlling of the spectrum. 
Second, we consider the map $D_L: \mathbb{R}^2  \rightarrow \mathbb{R}$ defined by $(\alpha, \beta) \mapsto Discr\left(\chi_{\alpha L + \beta \underline{L}\left( \mathscr{E}\right)}\right)$, and we apply a classical algebraic variety result on non identically null polynomial maps in several variables: the sets of zeros of such maps are of Lebesgue measure zero, and their complement is dense (see \cite{Federer}). This permits us to conclude.

The proof of Theorem \ref{two} is slightly more involved. {Starting from a Laplacian $L$ over a graph with $n$ nodes and the corresponding eigenvector equation $L v = \lambda v$, we decompose $v = (X,z)$ and $L$ in blocks matrices, one of these blocks being given by the Laplacian over the first $n-1$ nodes of the initial graph.} Perturbing $L$ we can make the real number $z$ distinct from zero. Then we obtain an  equation of the form $U X = b$, where $U$ is an operator depending on the perturbation and on the graph. We invert this operator and show that all entries of $X$ are polynomials in the  eigenvalue. Finally, using this observation the conclusion follows from the regularity of algebraic varieties.

\section{Structural genericity of graph Laplacians with simple spectrum: Proof of Theorem \ref{one}}\label{proof1}

Here we deal with Question \textbf{(Q1)} for undirected graphs. In this case the Laplacian matrices are symmetric and their eigenvalues are real. We aim at proving that having simple eigenvalues is structurally generic in $\mathcal{W}$.

\begin{proof}[Proof of Theorem \ref{one}]  \textbf{(a)} First we prove that $\mathcal{W}_s$ is structurally dense in the set $\mathcal{W}$.  Let us fix an element $L$ of $\mathcal{W}$, and $\mathcal{G}$ its associated weighted, undirected, connected graph.  The idea of the proof is to construct a graph with the same structure as $\mathcal{G}$ and simple spectrum. We will do this iteratively: we start from a path, then we will include step-by-step nodes and edges to recover the topology of $\mathcal{G}$. At each step of the process, we will impose smaller and smaller weights in the edges added, so as to keep a simple spectrum.

As $\mathcal{G}$ is connected, it admits a spanning tree $\mathcal{T}$. Then let's consider the longest path $\mathcal{P}$ in $\mathcal{T}$ (see Figure \ref{Figure1}). Without loss of generality we can reorder if necessary the nodes of $\mathcal{G}$ so that $\mathcal{T}$ is rooted at the node $1$ and $\mathcal{P}=\{1,2,\cdots,p-1,p\}$ where $p\leq n$: indeed, reordering the nodes of $\mathcal{G}$ thanks to a permutation $\sigma$ comes to considering the matrix $P_{\sigma}^{-1}LP_{\sigma}$ instead of $L$ (where $P_{\sigma}$ stands for the permutation matrix defined by $\sigma$), which does not restrict the generality since two similar matrices have the same spectra.

For any tuple $\mathscr{A}=\left(a_{1,2},\cdots,a_{p-1,p}\right)$ in ${\Er_+^*}^{p-1}$, the Laplacian matrix $L_\mathcal{P}\left(\mathscr{A}\right)$ (of the path $\mathcal{P}$) weighted by $\mathscr{A}$ has the following form:
\begin{align*}
\small{
L_\mathcal{P}\left(\mathscr{A}\right)=\begin{bmatrix}
a_{1,2} & -a_{1,2} &0 &\cdots &\cdots &\cdots &0\\
-a_{1,2}&a_{1,2}+a_{2,3}& -a_{2,3}&0&\cdots &\cdots &0\\
0 & -a_{2,3} & a_{2,3}+a_{3,4} & -a_{3,4} &0 &\cdots & 0\\
\quad \\
\vdots & &\ddots & \quad \ddots &&\ddots & &\\
\quad \\
&&&&&&0\\
0 & \cdots & \cdots& 0 &-a_{p-2,p-1}&a_{p-2,p-1}+a_{p-1,p}& -a_{p-1,p}\\
0 & \cdots & \cdots& 0 &0 &-a_{p-1,p}& a_{p-1,p}\\
\end{bmatrix}.
}
\end{align*}

The spectrum of such tridiagonal matrices of size $p$ has been intensively studied. In particular, it is well known that if $\mathscr{A}$ is only composed of two-by-two distinct reals, $L_\mathcal{P}\left(\mathscr{A}\right)$ has only simple eigenvalues. Let us take one such tuple $\mathscr{A}$ with positive distinct two-by-two elements, we can then write:
\begin{align*}
\mathfrak{S}\left(L_\mathcal{P}\left(\mathscr{A}\right)\right)=\{0<\alpha_2< \cdots <\alpha_p \},
\end{align*}
where the positive real numbers $\alpha_i$ depend on the $a_i$.\\
\begin{figure}[h]
 \centering
\includegraphics[scale=0.5]{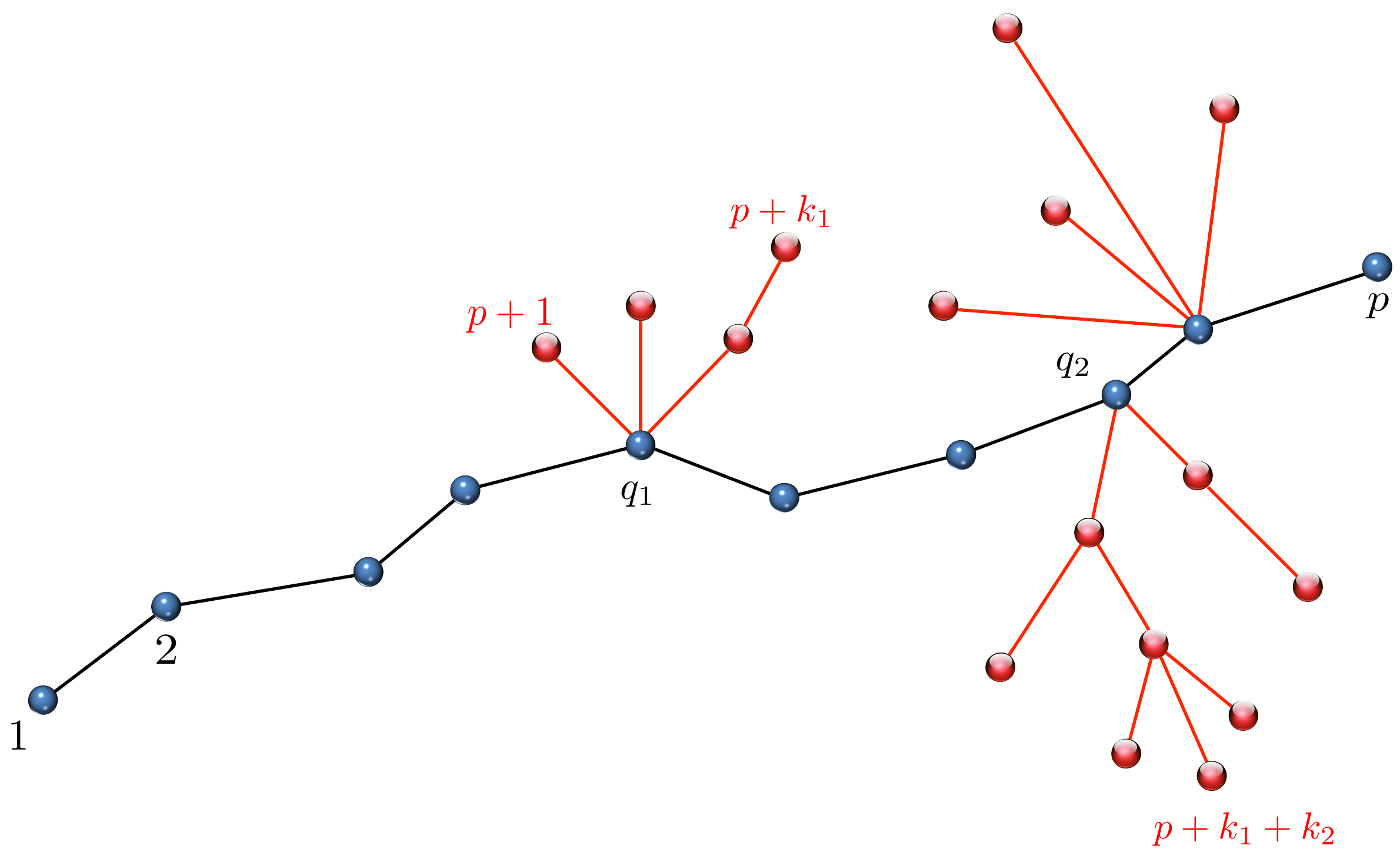}\\
\caption{The path $\mathcal{P}$ (in black) inside the spanning tree $\mathcal{T}$, both of them subgraphs of $\mathcal{G}$. We have represented in red the out-branches of $\mathcal{T}$, notably $\mathcal{B}_1$ and $\mathcal{B}_2$ rooted at nodes $q_1$ and $q_2$.}
\label{Figure1}
\end{figure}

Now let us consider the subgraph $\mathcal{P}\cup\{p+1,\cdots,n\}$, i.e the disconnected graph composed of the path $\mathcal{P}$ and of the other nodes of $\mathcal{G}$ (Figure \ref{Figure1}). Then the Laplacian matrix of this subgraph is the following square matrix of size $n$:
\begin{align*}
L_{\mathcal{P}\cup\{p+1,\cdots,n\}}\left(\mathscr{A}\right)&=
\left(\begin{array}{c|c}
\\
 L_\mathcal{P}\left(\mathscr{A}\right)
&\left(0\right)\\ 
 \\
\hline 
 \\
\left(0\right) &\left(0\right)\\
\end{array}\right),
\intertext{for which the spectrum is:}
\mathfrak{S}\left(L_{\mathcal{P}\cup\{p+1,\cdots,n\}}\left(\mathscr{A}\right)\right)&=\{0^{\otimes\left(n-p+1\right)}< \alpha_2<\cdots<\alpha_p\}.
\end{align*}
Now, we look at the edges of the tree $\mathcal{T}$ that do not belong to the path $\mathcal{P}$: let us consider the first of these edges, starting from the root $1$ of this tree. 
This first edge is adjacent to a node say $q_1$ in $\mathcal{P}$, and without loss of generality we can assume it is the edge $\mathcal{E}_1=\left(q_1,p+1\right)$. Re-ordering the nodes if necessary we can assume the first out-branch $\mathcal{B}_1$ starting at this edge $\mathcal{E}_1$ is composed of the nodes $q_1,p+1,\cdots, p+k_1$, where $k_1+1$ is the number of nodes of $\mathcal{B}_1$ (Figure \ref{Figure1}).
Then, let's take $\epsilon_1>0$, small enough compared to the $a_i$: the subgraph $\mathcal{P}\cup \mathcal{E}_1\cup \{p+2,\cdots,n\}$ with weights given by the tuple $\left(\mathscr{A}, \epsilon_1\right)$ has one component less than the graph $\mathcal{P}\cup\{p+1,\cdots,n\}$, and satisfies the relation:
\begin{align*}
L_{\mathcal{P}\cup \mathcal{E}_1\cup \{p+2,\cdots,n\}}\left(\mathscr{A},\epsilon_1\right)&=L_{\mathcal{P}\cup\{p+1,\cdots,n\}}\left(\mathscr{A}\right)+ L_{\{1,\cdots,p\}\cup\mathcal{E}_1\cup\{p+2,\cdots,n\}}\left(\epsilon_1\right).
\intertext{Therefore the spectrum of $\mathcal{P}\cup \mathcal{E}_1\cup \{p+2,\cdots,n\}$ has the following form:}
\mathfrak{S}\left(L_{\mathcal{P}\cup \mathcal{E}_1\cup \{p+2,\cdots,n\}}\left(\mathscr{A},\epsilon_1\right)\right)&=\{0^{\otimes\left(n-p\right)}<\beta_1\left(\epsilon_1\right)< \alpha_{2,1}<\cdots<\alpha_{p,1}\},
\end{align*}
where $\beta_1\left(\epsilon_1\right)$ is a small perturbation of $0$, and the $\alpha_{i,1}$ are small perturbations of the $\alpha_{i}$.\\
Now we repeat this process with the second edge $\mathscr{E}_2$ of the branch $\mathscr{B}_1$.
Let us consider the subgraph $\mathcal{P}\cup \mathcal{E}_1\cup{\mathcal{E}_2}\cup \{p+3,\cdots,n\}$
weighted by the tuple $\left(\kappa\left(\epsilon_1\right)[\mathscr{A},\epsilon_1], \epsilon_2\right)$
where $\kappa\left(\epsilon_1\right)>0$ is chosen large enough so that the eigenvalues $\kappa\left(\epsilon_1\right)\beta_1\left(\epsilon_1\right), \kappa\left(\epsilon_1\right)\alpha_{2,1},\cdots,\kappa\left(\epsilon_1\right)\alpha_{p,1}$ are very
large numbers compared to $0$, and where $\epsilon_2>0$ is small enough compared to the entries of $\kappa\left(\epsilon_1\right)[\mathscr{A},\epsilon_1]$.\\ Then this subgraph $\mathcal{P}\cup \mathcal{E}_1\cup{\mathcal{E}_2}\cup \{p+3,\cdots,n\}$
 has one connected component less than the subgraph $\mathcal{P}\cup \mathcal{E}_1\cup \{p+2,\cdots,n\}$.\\
Therefore the spectrum $\mathfrak{S}\left(L_{\mathcal{P}\cup \mathcal{E}_1\cup\mathcal{E}_2\cup \{p+3,\cdots,n\}}\left(\kappa\left(\epsilon_1\right)[\mathscr{A},\epsilon_1],\epsilon_2\right)\right)$ is of the form:
\begin{align*}
\{0^{\otimes\left(n-p-1\right)}<\beta_2\left(\epsilon_2\right)<\beta_{1,2}< \alpha_{2,2}<\cdots<\alpha_{p,2}\},
\end{align*}
where $\beta_2\left(\epsilon_2\right)$ is a small perturbation of $0$, and the terms $\beta_{1,2}$, $\alpha_{i,2}$ are small perturbations of the real numbers $\beta_1\left(\epsilon_1\right)$, $\alpha_{i,1}$.\\
Repeating this process for the rest of the nodes $p+3,\cdots,p+k_1$ of $\mathcal{B}_1$ and for the other out-branches $\mathcal{B}_2$ (rooted at a node say $q_2$), $\mathcal{B}_3,\cdots$ of the tree $\mathcal{T}$, we get the existence of a tuple $\mathscr{E}$ in ${\Er_+^*}^{n-1}$ for which the spectrum of the Laplacian matrix $L_{\mathcal{T}}\left(\mathscr{E}\right)$ with weights given by $\mathscr{E}$ 
has only simple eigenvalues:
\begin{align*}
 \mathfrak{S}\left(L_{\mathcal{T}}\left(\mathscr{E}\right)\right)=
\{0<\lambda_2\left(\mathscr{E}\right)<\cdots<\lambda_n\left(\mathscr{E}\right)\}.
\end{align*}

We can now complete the tree $\mathcal{T}$ with edges weighted by a tuple $\mathscr{E}'$ in $\Er_+^{{\frac{n \left(n-3\right)}{2}}}$, so as to recover the complete graph $\mathcal{G}$ weighted by a tuple $\mathscr{E}_0=\left(\mathscr{E},\mathscr{E}'\right) \in \Er_+^{{\frac{n \left(n-1\right)}{2}}}$: choosing the real numbers of $\mathscr{E}'$ very small compared to $\mathscr{E}$, we get that the Laplacian matrix $\underline{L}\left(\mathscr{E}_0\right)$ weighted by $\mathscr{E}_0$ belongs to $\mathcal{W}_s$.

Finally, to get our structural density result, it just remains to apply a  classical argument on algebraic varieties. Indeed, consider the map $D_L$ defined by:
\begin{align*}
\begin{array}{ccccc}
D_L& : & \Er^2 & \to & \Er\\
 & & \left(\alpha,\beta\right) & \mapsto & {Discr}\left({\chi}_{\alpha L+\beta \underline{L}\left(\mathscr{E}_0\right)}\right) \\
\end{array},
\end{align*}
where ${Discr}$ stands for the Discriminant map over the field $\Er_n [X]$ of real polynomials of degree $n$. Recall the map ${Discr}$ is itself a polynomial map (in the coefficients of the element of $\Er_n [X]$ considered) and that it satisfies the following relation:
\begin{align*}
\forall P \in \Er_n [X]:\,\, {Discr} \left(P\right)=c\left(P\right)\prod_{i<j}\left(\alpha_i-\alpha_j\right)^2,
\end{align*}
where $c\left(P\right)$ is a constant and the $\alpha_i$ are the roots of the polynomial $P$. (Notice this formula gives us that ${Discr} \left(P\right)=0$ if and only if $P$ admits at least one multiple eigenvalue).\\
Then $D_L$ is a polynomial map in the entries of the matrix $\alpha L+\beta \underline{L}\left(\mathscr{E}_0\right)$. Moreover, this map $D_L$ is not identically null over $\Er^2$ since we have proved above that $D_L\left(0,1\right)\neq 0$.\\
Therefore the complement of the algebraic variety $D_L^{-1}\left(\{0\}\right)$ is dense in $\Er^2$. Thus we have:
\begin{align*}
\forall a >0,\,\,  \exists \,\,0<\beta <\alpha \text{ such that } \dfrac{\beta}{\alpha}<a \text{ and } L+\dfrac{\beta}{\alpha}\underline{L}\left(\mathscr{E}_0\right) \in \mathcal{W}_s.
\end{align*}
We have proved that $\mathcal{W}_s$ is structurally dense in $\mathcal{W}$.\\

\textbf{(b)} Now let's fix $L$ in $\mathcal{W}$ and prove the set of $\mathscr{E}$ in $\Er^{{\frac{n \left(n-1\right)}{2}}}$
such that $\underline{L}\left(\mathscr{E}\right)$ belongs to ${\mathcal{W}_s}^c\cap \mathcal{W}$ is of Lebesgue measure zero in $\Er^{{\frac{n \left(n-1\right)}{2}}}$. Actually the proof of this fact is already given by the first point we established above.\\
Indeed, it suffices to consider the map:
\begin{align*}
\begin{array}{ccccc}
E_L& : & \Er^{{\frac{n \left(n-1\right)}{2}}}& \to & \Er\\ 
 & & \mathscr{E}& \mapsto & {Discr}\left({\chi}_{\underline{L}\left(\mathscr{E}\right)}\right)
\end{array},
\end{align*}
that we proved to be a non-null polynomial map. As before, we conclude the algebraic variety $E_L^{-1} \left(0\right)$ is of Lebesgue measure zero in $\Er^{{\frac{n \left(n-1\right)}{2}}}$. QED. 
\end{proof}


\section{Structural genericity of graph Laplacians having a Fiedler vector with non-zero entries: Proof of Theorem \ref{two}}\label{proof2}

In this section we tackle Question \textbf{(Q2)}, which 
asks how big the complement set ${\mathcal{W}_0}^c\cap \mathcal{W}$ is in the set $\mathcal{W}$. 

\begin{proof}[Proof of Theorem \ref{two}]
\textbf{(a)} First let us prove that ${\mathcal{W}_0}^c\cap \mathcal{W}$ is structurally dense in $\mathcal{W}$.\\
Let us fix $L$ in $\mathcal{W}$ a Laplacian matrix of a graph $\mathcal{G}$, and $L_{n-1}$ the Laplacian matrix of the subgraph of $\mathcal{G}$ induced by the first $n-1$ nodes.\\
As $\mathcal{W}_s$ is structurally dense in $\mathcal{W}$, we can assume that $L$ is in $\mathcal{W}_s$.
Denote by $\left(\alpha,a\right)$ the non identically null tuple in $\Er_+^{\frac{\left(n-1\right)\left(n-2\right)}{2}}\times \Er_+^{n-1}$ defining the weights of $L$, where $\alpha$ is the tuple in $\Er_+^{\frac{\left(n-1\right)\left(n-2\right)}{2}}$ defining the weights of $L_{n-1}$. The matrix $L$ has the following form:

\begin{equation}\label{shape}
L=
\left(\begin{array}{ccccc|c}
\\
&& &&&-a_1\\
&&&&& \vdots \\
&&L_{n-1}+D\left(a\right)&&& -a_i\\
&&&&&\vdots \\
&&&&&-a_{n-1} \\
&&&&&\\
\hline 
&&&&&\\
-a_1 &\cdots &-a_i&\cdots &-a_{n-1}&\sum_{k=1}^{n-1}a_k\\
\end{array}\right),
\end{equation}
where the matrix $D\left(a\right)$ is the diagonal matrix of size $n-1$ containing the elements of the non identically null tuple $a=\left(a_1,\cdots,a_{n-1}\right)$.
In accordance with our definition of structural density, we are going to prove that, up to a perturbation of the weights of $L$, i.e of the non-zero parameters in $\left(\alpha,a\right)$, the matrix $L$ belongs to $\mathcal{W}_0$.\\

Consider the algebraic connectivity $\lambda_2=\lambda_2 \left(\alpha,a\right)$ of $L$
: for a generic choice of the non-zero parameters in the tuple $\left(\alpha,a\right)$, the real number $\lambda_2$ does not belong to the spectrum of $L_{n-1}$ (this fact is proved in Lemma \ref{appendix} in the Appendix). Therefore perturbing our non-null parameters if necessary, we can assume $\lambda_2$ is not in this spectrum. 

For an eigenvector $\left(X,z\right)=\left(X\left(\alpha,a\right),z\left(\alpha,a\right)\right)$ in $\Er^n$ associated to $\lambda_2$ we have:
\begin{align*}
\begin{bmatrix}
\\
\\
L_{n-1}\left(\alpha \right)X\\
\\
0\\
\end{bmatrix}
+\begin{bmatrix}
a_1\left(X_{\left(1\right)}-z\right)\\
\vdots \\
a_{n-1}\left(X_{\left(n-1\right)}-z\right)\\
&\\
-\sum_{k=1}^{n-1}a_k\left(X_{\left(k\right)}-z\right)\\
\end{bmatrix}
=\lambda_2 \begin{bmatrix}
\\
\\
X\\
\\
z\\
\end{bmatrix}.
\end{align*}

(i) Now, assume that in the tuple $a$, only the real numbers $a_1$ and $a_i$ are non zero (this case does not restrict the generality as is shown later in the proof).\\
In this case the eigenvector equation above can be written as:
\begin{align}\label{Evector}
\begin{bmatrix}
\\
L_{n-1}\left(\alpha \right)X\\
\\
\end{bmatrix}
&+\begin{bmatrix}
a_1\left(X_{\left(1\right)}-z\right)\\
\mathbf{\left({0}\right)}\\
a_{i}\left(X_{\left(i\right)}-z\right)\\
\mathbf{\left({0}\right)}
\end{bmatrix}
=\lambda_2 \begin{bmatrix}
\\
X\\
\\
\end{bmatrix}\\
\intertext{and:}
a_i\left(X_{\left(i\right)}-z\right)&=-a_1\left(X_{\left(1\right)}-z\right)-\lambda_2 z.
\end{align}
\textbf{1st case: Assume we have $X_{\left(i\right)} \neq z$.\\}
In this case, the theory of perturbations of eigenvalues (see \cite{Watkins, Sun1990}) tells us that the algebraic connectivity map $x \mapsto \lambda_2 \left(\alpha,a_1,x\right)$ is differentiable in $x=a_i$ (since $\lambda_2 \left(\alpha,a_1,a_i\right)$ is simple) and its derivative is given by:
\begin{align*}
{\dfrac{\partial}{\partial x}}_{\vert{x=a_i}}\lambda_2 \left(\alpha,a_1,x\right)= \dfrac{{}^{\textbf{t}}
\begin{bmatrix}
X\\
z \\
\end{bmatrix}
\cdot L\left(0,\cdots,1,0,\cdots,0\right)\cdot \begin{bmatrix}
X\\
z \\
\end{bmatrix}}{\vert \vert \begin{bmatrix}
X\\
z \\
\end{bmatrix}\vert \vert ^2}=\left(X_{\left(i\right)}-z\right)^2>0.
\end{align*}
Thus, perturbing the parameter $a_i$ if necessary, we have:
\begin{align*}
\lambda_2=\lambda_2 \left(\alpha,a_1,a_i\right) \notin \mathfrak{S}\left( L_{n-1}\left(\alpha \right)+ M\left(a_1\right)\right),
\intertext{where:}
M\left(a_1\right)=\begin{bmatrix}
a_1& 0& \cdots&0\\
\mathbf{\left({0}\right)}& \mathbf{\left({0}\right)}& \cdots&\mathbf{\left({0}\right)}\\
-a_1& 0& \cdots&0\\
\mathbf{\left({0}\right)}& \mathbf{\left({0}\right)}& \cdots&\mathbf{\left({0}\right)}\\
\end{bmatrix}.
\end{align*}
Under this possible small perturbation of $a_i$, the new algebraic connectivity $\lambda_2 \left(\alpha,a_1,a_i\right)$ and the new Fiedler vector $\left(X,z\right)$ still satisfies $X_{\left(i\right)}\neq z$. From Equation \eqref{Evector} we then obtain the following relation:
\begin{align*}
\left(L_{n-1}\left(\alpha \right)+ M\left(a_1\right)-\lambda_2 I_{n-1}\right)X=\begin{bmatrix}
a_1z\\
\mathbf{\left({0}\right)}\\
\lambda_2 z -a_1z\\
\mathbf{\left({0}\right)}\\
\end{bmatrix},
\end{align*}
in which only $\lambda_2$ depends on the parameter $a_i$. From this we get $z\neq 0$ and:
\begin{align*}
X&=\dfrac{1}{\mu\left(\alpha, a_1,a_i\right)}\,
{}^{\textbf{t}} \textbf{Com}\left[L_{n-1}\left(\alpha \right)+
M\left(a_1\right)
-\lambda_2 I_{n-1}\right]
\begin{bmatrix}
a_1z\\
\mathbf{\left({0}\right)}\\
\lambda_2 z -a_1z\\
\mathbf{\left({0}\right)}\\
\end{bmatrix}\\
&=\dfrac{1}{\mu\left(\alpha, a_1,a_i\right)}\,
\textbf{Com}\left[L_{n-1}\left(\alpha \right)+
{}^{\textbf{t}}M\left(a_1\right)
-\lambda_2 I_{n-1}\right]
\begin{bmatrix}
a_1z\\
\mathbf{\left({0}\right)}\\
\lambda_2 z -a_1z\\
\mathbf{\left({0}\right)}\\
\end{bmatrix}
\end{align*}
where $\mu\left(\alpha, a_1,a_i\right)$ denotes the determinant of the matrix $L_{n-1}\left(\alpha \right)+
M\left(a_1\right)-\lambda_2 I_{n-1}$. 
We are thus led to look at the comatrix of $L_{n-1}\left(\alpha \right)+
{}^{\textbf{t}}M\left(a_1\right)-\lambda_2 I_{n-1}$. Let us denote by $\left(c_{i,j}\right)_{1\leq i,j\leq n-1}$ the entries of this comatrix and let us first compute all the coefficients $c_{k,1}$: to do this it suffices to develop with respect to the first column in the matrix $L_{n-1}\left(\alpha \right)+{}^{\textbf{t}}M\left(a_1\right)-\lambda_2 I_{n-1}$. 
We get:
\begin{align*}
c_{11}&= \prod_{j=1}^{n-2} \left(\lambda'_j\left(1,1\right)-\lambda_2 \right),\\
\forall k>1,\,\, c_{k1}&=\left(-1\right)^{k+1}\left(\prod_{j=1}^{n-2} \left(\lambda'_j\left(k,1\right)-\lambda_2 \right)-a_1\left(-1\right)^{i+1}\prod_{j=1}^{n-3} \left(\lambda''_j\left(k,1,1,i\right)-\lambda_2 \right)\right),
\end{align*}
where $\lambda'_j\left(1,1\right)$ denote the eigenvalues of the submatrix $L_{n-1}\left(1\vert1\right)$ of size $n-2$, obtained by deleting the $1$th row and the $1$th column in $L_{n-1}\left(\alpha\right)$, $\lambda'_j \left(k,1\right)$ denote the eigenvalues of the submatrix $L_{n-1}\left(k\vert 1\right)$ of size $n-2$, obtained by deleting the $k$th row and the $1$th column, and where the terms $\lambda''_j \left(k,1,1,i\right)$ denote the eigenvalues of the submatrix (of size $n-3$) of $L_{n-1}\left(k\vert1\right)$, obtained by deleting the $1$th row and the $i$th column in $L_{n-1}\left(k\vert1\right)$. With the same notations, we have:
\begin{align*}
c_{1i}&= \left(-1\right)^{1+i}\prod_{j=1}^{n-2} \left(\lambda'_j\left(1,i\right)-\lambda_2 \right),\\
\forall k>1,\,\, c_{ki}&=\left(-1\right)^{k+i}\left(\prod_{j=1}^{n-2} \left(\lambda'_j\left(k,i\right)-\lambda_2 \right)+a_1\prod_{j=1}^{n-3} \left(\lambda''_j\left(k,i,1,1\right)-\lambda_2 \right)\right).
\end{align*}
We notice that all the eigenvalues involved above $\lambda'_j \left(1,1\right)$, $\lambda'_j \left(k,1\right)$, $\lambda'_j \left(1,i\right)$, $\lambda'_j \left(k,i\right)$ and $\lambda''_j \left(k,1,1,i\right), \lambda''_j \left(k,i,1,1\right)$ are eigenvalues of submatrices of $L_{n-1}\left(\alpha \right)$: they do not depend on the coefficient $a_i$, and thus changing slightly $a_i$ if necessary, we have that (for a generic choice of the parameter $a_i$) all those eigenvalues are distinct from the algebraic connectivity $\lambda_2=\lambda_2 \left(\alpha, a_1,a_i\right)$.

Therefore, for every integer $1\leq k \leq n$ the coefficients $c_{k1}$ and $c_{ki}$ are (non identically null) polynomials of degree $n-2$ in the eigenvalue $\lambda_2=\lambda_2 \left(\alpha,a_1,a_i\right)$, each of the coefficients of these polynomials expressions being independent on the coefficient $a_i$.
Thus we conclude that the Fiedler vector $X$ has the form:
\begin{align*}
X=\dfrac{z}{\mu\left(\alpha, a_1,a_i\right)}\begin{bmatrix}
P_1\left( \lambda_2 \right)\\
\vdots \\
P_{n-1}\left( \lambda_2 \right)
\end{bmatrix},
\end{align*}
Where each $P_i$ are non identically null polynomials of degree $n-1$ in the eigenvalue $\lambda_2$. Therefore, we get that, up to a perturbation of the (non-null) parameter $a_i$ the eigenvector $\left(X,z\right)$ has only non-zero components.\\

\noindent \textbf{2nd case: Assume we have $X_{\left(i\right)}=z$.\\}
Then we have $a_1 \left(X_{\left(1\right)}-z\right)= -\lambda_2 z$, with $X_{\left(1\right)} \neq z$, otherwise : $\lambda_2 \in \mathfrak{S}\left(L_{n-1}\left(\alpha \right)\right)$. Thus $z\neq 0$ and:
\begin{align*}
X=\dfrac{1}{\mu\left(\alpha,a_1,a_i\right)}\,
{}^{\textbf{t}} \textbf{Com}\left[L_{n-1}\left(\alpha \right)-\lambda_2 I_{n-1}\right]
\begin{bmatrix}
\lambda_2 z\\
\mathbf{\left({0}\right)}\\
\end{bmatrix}\\
=\dfrac{1}{\mu\left(\alpha,a_1,a_i\right)}\,
 \textbf{Com}\left[L_{n-1}\left(\alpha \right)-\lambda_2 I_{n-1}\right]
\begin{bmatrix}
\lambda_2 z\\
\mathbf{\left({0}\right)}\\
\end{bmatrix}.
\end{align*}
Now the derivative of the map $x \mapsto \lambda_2 \left(\alpha,x,a_i\right)$ satisfies:
\begin{align*}
{\dfrac{\partial}{\partial x}}_{\vert{x=a_1}}\lambda_2 \left(\alpha,x,a_i\right)=\left(X_{\left(1\right)}-z\right)^2>0,
\end{align*}
so, perturbing $a_1$ if necessary, we have this time the relation 
\[c_{k1}=\left(-1\right)^{k+1} \prod_{j=1}^{n-2} \left(\lambda'_j\left(k,1\right)-\lambda_2 \right) \neq 0.
\] 
Under this small perturbation of $a_1$, either the new eigenvalue $\lambda_2$ and the new eigenvector $\left(X,z\right)$ satisfy $X_{\left(i\right)} \neq z$ in which case we are led to the 1st case, or they satisfy $X_{\left(i\right)} =z$ in which case we get again $z\neq 0$ and the relation:
\begin{align*}
X=\dfrac{z}{\mu\left(\alpha, a_1,a_i\right)}\begin{bmatrix}
c_{1,1} \lambda_2 \\
\vdots \\
c_{n-1,1}\lambda_2 \\
\end{bmatrix},
\end{align*}
which implies the eigenvector $\left(X,z\right)$ has only non-zero coordinates.\\
So we proved that, up to a perturbation of the non-null parameters in the tuple $\left( \alpha,a_1,a_i \right)$, the Fiedler vector $\left(X \left(\alpha,a_1,a_i\right),z\left(\alpha,a_1,a_i\right)\right)$ has only non-zero coordinates.

(ii) The general situation where $a$ is not identically null is completely similar. Indeed, either there is only one non-zero coefficient say $a_i$ (in which case the situation is easier than in (i), since all the entries $c_{i,j}$ of the comatrix appeared above are equal in absolute value to $\vert \prod_{j=1}^{n-2} \left(\lambda'_j\left(k,l\right)-\lambda_2 \right)\vert$), or more than two of them are distinct from zero: in this former case the exactly same reasoning as in (i) applies, from which we get again that the coordinates of the vector $X$ can be expressed as non-null polynomials in the eigenvalue $\lambda_2$.

Thus in any case, for a generic choice of the non-null parameters in $\left(\alpha,a\right)$, i.e for a generic choice of the non-null weights of our initial Laplacian matrix $L$, the Fiedler vector $\left(X\left(\alpha,a\right),z\left(\alpha,a\right)\right)$ associated to the algebraic connectivity $\lambda_2\left(\alpha,a\right)$ has only non-zero components.\\
In other words, for any Laplacian matrix $L$ in $\mathcal{W}$, there exists a tuple $\mathscr{E}$ in $\Er_+^{\frac{n\left(n-1\right)}{2}}$ of norm as small as we want, such that the matrix $L+{\underline{L}}\left(\mathscr{E}\right)$ belongs to $\mathcal{W}_0$. QED.\\

\textbf{(b)} Secondly, let us fix $L$ in $\mathcal{W}$ and prove the set of tuples $\mathscr{E}$ in $\Er^{\frac{n\left(n-1\right)}{2}}$ for which $\underline{L}\left(\mathscr{E}\right)$ belongs to $\mathcal{W}_0$ is of Lebesgue measure zero.\\

Let $L$ a Laplacian matrix in $\mathcal{W}_0$ with weights given by the tuple $\left(\alpha,a\right)$, and denote by $\lambda_2=\lambda_2\left(\alpha,a\right)$ its second eigenvalue and by $X=X\left(\alpha, a\right)$ its Fiedler vector. Then in the same notations as in the proof of structural density we have just done above, either we have:
\begin{align*}
\lambda_2 \in \mathfrak{S}\left(L_{n-1}\right)&\cap \mathfrak{S}\left(L\right)
\intertext{which implies:}
\mathfrak{S}\left(L_{n-1}\right)&\cap \mathfrak{S}\left(L\right) \neq \{0\},
\intertext{or $\lambda_2$ does not belong to both of these spectra: in this last case the algebraic connectivity $\lambda_2 \left(\alpha, a\right)$ is a non constant map in one of the parameters of the tuple $\left(\alpha,a\right)$, say $a_i$, and we have:}
X&=\begin{bmatrix}
Q_1\left( \lambda_2 \right)\\
\vdots \\
Q_{n-1}\left( \lambda_2 \right)
\end{bmatrix}, 
\end{align*}
with $Q_i\left( \lambda_2 \right)=0$ for at least one index $i$, where the $Q_i$ are one variable polynomials of degree $n-1$, of which coefficients depend on all the parameters of the tuple $\left(\alpha,a\right)$ except $a_i$.
Thus we have:
{\small
\begin{eqnarray}
\mathscr{L}_{\vert \Er^{\frac{n\left(n-1\right)}{2}}}\left(\{\underline{L}\left(\mathscr{E}\right)\in \mathcal{W}_0\}\right)
&\leq &\mathscr{L}_{\vert \Er^{\frac{n\left(n-1\right)}{2}}}\left(\{\left(\mathscr{E}_1,\mathscr{E}_2\right) \in \Er^{\frac{n\left(n-1\right)}{2}}: \mathfrak{S}\left(\underline{L}_{n-1}\left(\mathscr{E}_1\right)\right)\cap \mathfrak{S}\left(\underline{L}\left(\mathscr{E}_1,\mathscr{E}_2\right)\right)\neq \{0\}\}\right) \nonumber \\ 
&\quad + & \mathscr{L}_{\vert \Er}\left(\cup_{i=1}^n\text{Zeros}\left(Q_i\right)\right). \nonumber 
\end{eqnarray}}
It is actually proven in Lemma \ref{appendix} of the Appendix that the first term of this sum is $0$. And as a set of zeros of any non null polynomial in one variable is a finite set in $\Er$, therefore:
\begin{align*}
&\mathscr{L}_{\vert \Er^{\frac{n\left(n-1\right)}{2}}}\left(\{\underline{L}\left(\mathscr{E}\right)\in \mathcal{W}_0\}\right)=0,
\end{align*}
as desired.
\end{proof}

The reader might have observed that the proof of Theorem \ref{two} can be applied to all non-zero simple eigenvalues of a Laplacian matrix in $\mathcal{W}$, and not only to the algebraic connectivity. And as the eigenvector of the eigenvalue $0$ is $\left(1,\cdots,1\right)$, the following corollary holds:
\begin{coro}\label{cor1}
The property of having a basis of eigenvectors with only non-zero entries is structurally generic in $\mathcal{W}$. In other words, 
denoting by $\mathcal{W}_a$ the subset of $\mathcal{W}$ formed by Laplacian matrices for which there exists a basis of eigenvectors, one of them admitting at least one zero component, we have:\\
 - The set ${\mathcal{W}_a}^c\cap \mathcal{W}$ is structurally dense in the set $\mathcal{W}$.\\
 - For any $L$ in $\mathcal{W}$, the set of tuples $\mathscr{E}$ in $\Er^{{\frac{n \left(n-1\right)}{2}}}$ such that 
$\underline{L}\left( \mathscr{E}\right)$ belongs to ${\mathcal{W}_a}$ is of Lebesgue measure zero in $\Er^{{\frac{n \left(n-1\right)}{2}}}$.
\end{coro}

\section{Structural genericity of graph Laplacians with simple spectrum: the directed case}\label{Digraphs}
Section \ref{proof1} was devoted to prove our first structural genericity result in the undirected case: actually 
Theorem \ref{one} extends to the case of directed graphs (i.e simple graphs where edges have been assigned an orientation), provided we consider directed graphs presenting a particular kind of ``connectedness". Let us first recall the notion of connectedness in the directed case.

\begin{defi}
A \textit{directed graph} (or in a short way a \textit{digraph}) $\mathcal{G}=\left(\mathcal{V},E\right)$ is {strongly connected}, if for any of its nodes $i$ and $j$, there exists a directed path between $i$ and $j$. Moreover, 
A digraph is said to be weakly connected if its underlying undirected graph, obtained by ignoring the orientations of the edges, is connected.
\end{defi}

Given a digraph, a \textit{directed rooted spanning tree} is a rooted spanning tree for which the edges are directed towards the root.
A \textit{diverging rooted spanning tree} is a rooted spanning tree for which the edges are directed away from the root.

As a consequence, there does not always exist a directed rooted spanning tree for a weakly connected digraph; nor does there always exist a diverging rooted spanning tree (see Figure \ref{toto}). Notice that other names are given in the literature for characterizing these objects: for instance diverging rooted spanning trees are sometimes called \textit{arborescence} or \textit{out-trees}, while directed rooted spanning trees are sometimes called \textit{anti-arborescence} or \textit{in-trees}.

\begin{nota}
We will denote by $\mathcal{W'}$ the set of Laplacian matrices over weakly connected weighted digraphs, and by $\mathcal{W'}_s$ the subset of $\mathcal{W'}$ of Laplacian matrices for which the eigenvalues are simple.
\end{nota}
As in the symmetric case, $0$ is always an eigenvalue for Laplacian matrices over directed graphs: moreover, $0$ is simple if and only if the underlying digraph admits a diverging spanning rooted tree. Other results can be found in the literature on the multiplicity of $0$ in terms of number of diverging spanning trees contained in diverging spanning forests of digraphs (see \cite{Jadbabaie2013, Chebotarev2009}).

Now, we aim at dealing with Question \textbf{(Q2)} for digraphs: actually Theorem \ref{one} cannot be extended to the whole set $\mathcal{W'}$: indeed if we take the directed graph $\mathcal{G}_1$ (with four vertices) given in Figure \ref{toto}, then $\mathcal{G}_1$ does not have a directed rooted spanning tree nor a diverging rooted spanning tree. Its corresponding Laplacian matrix is:
\begin{align*}
L_1=\begin{bmatrix}
0&0&0&0\\
-a_{2,1}&a_{2,1}+a_{2,4}&0&-a_{2,4}\\
-a_{3,1}&0&a_{3,1}+a_{3,4}&-a_{3,4}\\
0&0&0&0
\end{bmatrix},
\end{align*}
for which the eigenvalue $0$ is double, and this for any values of the non-null parameters $a_{i,j}$. Similarly for the graph $\mathcal{G}_2$ with five vertices (see Figure \ref{toto}) for which the eigenvalue $0$ of the corresponding Laplacian matrix $L_2$ is double as well.

\begin{figure}[h]
 \centering
\includegraphics[scale=0.4]{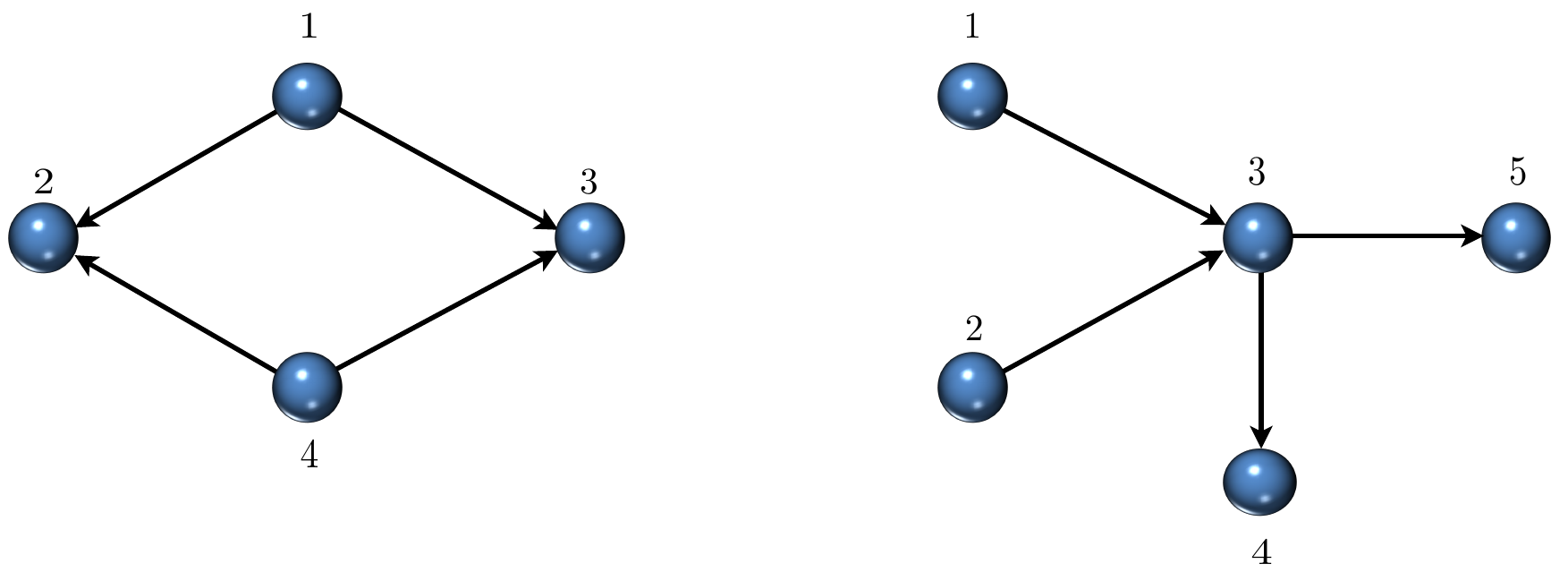}\\
\caption{Two weakly connected digraphs $\mathcal{G}_1$ (on the left), $\mathcal{G}_2$ (on the right): both do not have a directed spanning tree nor a diverging spanning tree.}
   \label{toto}
 \end{figure}
In fact our method used in the proof of Theorem \ref{one} works for digraphs which have a diverging rooted spanning tree.
\begin{nota}
We will denote by $\mathcal{U'}$ the subset of $\mathcal{W'}$ of Laplacian matrices over weighted digraphs having a diverging rooted spanning tree. We will denote by $\mathcal{U'}_s$ the subset of $\mathcal{U'}$ formed by Laplacian matrices with only simple eigenvalues: $\mathcal{U'}_s=\mathcal{U'} \cap \mathcal{W'}_s$
\end{nota}
\begin{defi}
For any matrix $M$ in $\mathcal{S}_n\left(\Er\right)^c \cap \mathcal{Z}_n\left(\Er\right)$ and any tuple $\mathscr{E}=\left(\epsilon_{i,j}\right)_{1\leq i\neq j\leq n}$ in $\Er^{n(n-1)}$, we define the matrix $\underset{\widetilde{}}{M}\left(\mathscr{E}\right)$ in $\mathcal{Z}_n\left(\Er\right)$ by:
\begin{align*}
\forall i\neq j, \,\,&\underset{\widetilde{}}{M}\left(\mathscr{E}\right)_{i,j}
=\begin{cases}
-\epsilon_{i,j} \text{ if } M_{i,j} \neq 0\\
0 \text{ else}\\
\end{cases}\\
\,\,&\underset{\widetilde{}}{M}\left(\mathscr{E}\right)_{i,i}=-\sum^n_{\substack{j=1\\
j \neq i}}\underset{\widetilde{}}{M}\left(\mathscr{E}\right)_{i,j}.
\end{align*}
\end{defi}
This definition holds only for non symmetric matrices $M$ in $\mathcal{S}_n\left(\Er\right)^c \cap \mathcal{Z}_n\left(\Er\right)$: indeed, in the case where $M$ is symmetric, 
it would not be interesting to consider a weighted matrix $\underset{\widetilde{}}{M}\left(\mathscr{E}\right)$ that would not be symmetric, for this would come to lose a structural information on $M$. Actually, for symmetric matrices $M$, the definition that we must consider is $\underline{M}\left(\mathscr{E}\right)$ (see definition \ref{defun}), which is the one fitting well with the undirected setting.

\begin{defi}
Let $\mathcal{M}_0\subset \mathcal{M}$ two subsets of $\mathcal{S}_n\left(\Er\right)^c \cap \mathcal{Z}_n\left(\Er\right)$. We say that $\mathcal{M}_0$ is \textit{structurally dense} in $\mathcal{M}$ if the following holds:
\[
\forall M \in \mathcal{M},\,\, \forall \epsilon_0 >0,\,\, \exists \mathscr{E} \in \Er^{n \left(n-1\right)},\,\,\text{such that } M+\underset{\widetilde{}}{M}\left( \mathscr{E}\right) \in\mathcal{M}_0 \text{ and } \lvert \lvert \underset{\widetilde{}}{M}\left( \mathscr{E}\right)\rvert \rvert< \epsilon_0.
\]
\end{defi}

From this definition of structural density, we deduce the corresponding notion of structural genericity for matrices in $\mathcal{S}_n\left(\Er\right)^c \cap \mathcal{Z}_n\left(\Er\right)$. The result now enunciates as:

\begin{thrm}\label{one'}
The property of having only simple eigenvalues is structurally generic in the set of Laplacian matrices $\mathcal{U'}$ over digraphs for which exists a diverging rooted spanning tree:\\
 - The set $\mathcal{U'}_s$ is structurally dense in the set $\mathcal{U'}$.\\
 - For any $L$ in $\mathcal{U'}$, the set of tuples $\mathscr{E}$ in $\Er^{{n \left(n-1\right)}}$ such that 
$\underset{\widetilde{}}{L}\left( \mathscr{E}\right)$ belongs to $\mathcal{U'}_s^c \cap \mathcal{U'}$ is of Lebesgue measure zero in $\Er^{n \left(n-1\right)}$.
\end{thrm}

\begin{proof}
The proof works in a completely similar as in the undirected case.\\
\textbf{(a)} First we prove that $\mathcal{U'}_s$ is structurally dense in the set $\mathcal{U'}$.
Let $L$ be an element of $\mathcal{U'}$, and $\mathcal{G}$ its associated weighted, weakly connected digraph.
 
By definition of $\mathcal{U'}$, $\mathcal{G}$ admits a diverging spanning tree $\mathcal{T}$, rooted at a node, say node $1$. Then let's consider the longest directed path $\mathcal{P}$ in $\mathcal{T}$, say $\mathcal{P}=\{1,2,\cdots,p-1,p\}$ where $p\leq n$. For any tuple $\mathscr{A}=\left(a_{2,1},\cdots,a_{p,p-1}\right)$ of distinct elements in ${\Er_+^*}^{p-1}$, the Laplacian matrix $L_\mathcal{P}\left(\mathscr{A}\right)$ (of the path $\mathcal{P}$) weighted by $\mathscr{A}$ is now lower triangular:
\begin{align*}
L_\mathcal{P}\left(\mathscr{A}\right)=\begin{bmatrix}
0 & 0 &0 &\cdots &\cdots &\cdots &0\\
-a_{2,1}&a_{2,1}& 0&0&\cdots &\cdots &0\\
0 & -a_{3,2} & a_{3,2}& 0 &0 &\cdots & 0\\
\quad \\
\vdots & &\ddots & \quad \ddots &&\ddots & &\\
\quad \\
&&&&&&0\\
0 & \cdots & \cdots& 0 &-a_{p-1,p-2}&a_{p-1,p-2}&0\\
0 & \cdots & \cdots& 0 &0 &-a_{p,p-1}& a_{p,p-1}\\
\end{bmatrix},
\end{align*}
and therefore it has a simple spectrum precisely equal to: $\mathfrak{S}\left(L_\mathcal{P}\left(\mathscr{A}\right)\right)=\{0,a_{2,1},\cdots,a_{p,p-1}\}$. 
Now, as the tree $\mathcal{T}$ is diverging, there exists an out-branch $\mathcal{B}_1$ starting at a node say $1\leq q_1\leq p$, composed of the nodes $q_1,p+1,\cdots, p+k_1$. The edges of $\mathcal{B}_1$ being directed away from $q_1$, the Laplacian matrix $L_{\mathcal{P}\cup \mathcal{B}_1\cup\{p+k_1+1,\cdots,n\}}\left(\mathscr{A},\mathscr{E}_1\right)$ is still lower triangular (and this for any tuple $\mathscr{E}_1$ in ${\Er_+^*}^{k_1}$). Therefore, its spectrum is directly given by the tuple $\left(\mathscr{A},\mathscr{E}_1\right)$, from which we deduce again that the non-zero eigenvalues of $L_{\mathcal{P}\cup \mathcal{B}_1\cup\{p+k_1+1,\cdots,n\}}\left(\mathscr{A},\mathscr{E}_1\right)$ are all distinct for a generic choice of $\mathscr{E}_1$ in ${\Er_+^*}^{k_1}$.\\
Repeating the process, we get at the end the existence of a tuple $\mathscr{E}_0 \in \Er_+^{n(n-1)}$, for which the Laplacian non symmetric matrix $\underset{\widetilde{}}{L}\left(\mathscr{E}_0\right)$ weighted by $\mathscr{E}_0$ belongs to $\mathcal{U'}_s$.
The rest of the proof is similar to Item \textbf{(a)} in the proof of Theorem \ref{one}. \\

\textbf{(b)} The exactly same reasoning applies, using the Discriminant map.
\end{proof}

\section{Conclusion and discussions}
We have defined and studied structural genericity properties for Laplacian matrices over weighted graphs, 
where only perturbations of non-zero entries were allowed, i.e, where the constraint of preserving the initial graph topology is imposed. Thanks to constructive methods, we proved that having simple spectrum and having a Fiedler vector with only non-zero components are both structurally generic for Laplacians of connected graphs.

According to a theorem by Fiedler on graph partitioning (see \cite{Fiedler1975}), given a connected undirected weighted graph $\mathcal{G}$, if the coordinates $v_i$ of the Fiedler vector $v$ are all distinct from zero, then the set $\mathcal{C}$ of edges $\left(i,j\right)$ for which  we have $v_iv_j<0$ forms a cut of $\mathcal{G}$ into two connected components, namely the set $\mathcal{G}_1$ of nodes $i$ for which $v\left(i\right)>0$ and the set $\mathcal{G}_2$ of nodes $i$ for which $v\left(i\right)<0$. In this case the partitioning of $\mathcal{G}$ corresponding to the cut $\mathcal{C}$ is uniquely determined (which is not always true) by $\mathcal{G}_1$, $\mathcal{G}_2$.

Therefore, our result given by Theorem \ref{two} adds a more precise information to this graph partitioning theorem: namely the fact that generically (in the structural meaning we have defined in this paper), the coordinates of the Fiedler vector always provide such a partition into two connected components. This refinement of the Fiedler's partitioning result  has  applications to the theory of synchronization of diffusively coupled dynamical systems. We will show in a future paper how we can use these genericity results to understand the effects of structural changes on the synchronizability of  complex systems.

Lastly, we conjecture that for a generic choice of the non-null weights of a given graph, the Fiedler vector has only distinct coordinates. This would be another extension of Theorem \ref{two}, useful as well for studying synchronization loss or synchronization enhancement in coupled networks. Such a direction will be followed in some other future works.

\subsection*{Acknowledgement}
We are in debt with Serhyi Yanchuk and Francisco Rodrigues for useful discussions.

\section*{Appendix}
Here, we prove the following result used in the proof of Theorem \ref{two} above, which asserts that given a Laplacian matrix $L$ over a connected undirected weighted graph $\mathcal{G}$, then for a generic choice of the weights of $L$, the non-null eigenvalues of this matrix are not eigenvalues of the Laplacian matrix of the subgraph of $\mathcal{G}$ formed by the nodes $1,\cdots,n-1$:

\begin{lemm}\label{appendix}
Let $L$ a Laplacian matrix in $\mathcal{W}$ over a graph $\mathcal{G}$ with $n$ nodes. Let $L_{n-1}$ the Laplacian matrix of the subgraph of $\mathcal{G}$ formed by the first $n-1$ nodes. Then, for a generic choice of the weights of $L$, the following holds:  
\begin{align*}
\mathfrak{S}\left({L}_{n-1}\right) \cap \mathfrak{S}\left(L\right) = \{0\}.
\end{align*}
More precisely, the set of parameters $\left(\mathscr{E}_1,\mathscr{E}_2\right)$ in $\Er^{\frac{\left(n-1\right)\left(n-2\right)}{2}} \times \Er^{n-1}$ for which we have $$\mathfrak{S}\left({\underline{L}}_{n-1}\left(\mathscr{E}_1 \right)\right) \cap \mathfrak{S}\left({\underline{L}}\left(\mathscr{E}_1,\mathscr{E}_2\right) \right) = \{0\}$$ is dense and its complement is of Lebesgue measure zero in 
$\Er^{\frac{n\left(n-1\right)}{2}}$.
\end{lemm}

Before we prove this lemma, we need the following auxiliary result:
\begin{lemm}\label{path}
Let $\mathcal{P}$ be a path over $p$ nodes. Then, there exists a tuple $b=\left(b_1,\cdots,b_{p-1}\right)$ in ${\Er_+^*}^{p-1}$ such that for the Laplacian matrix $L_{\mathcal{P}}\left(b\right)$ weighted by $b$, the eigenvalues are simple:
\begin{align*}
\mathfrak{S} \left(L_{\mathcal{P}}\left(b\right) \right)= \{0<\lambda_2\left(b\right)<\cdots< \lambda_{p}\left(b\right)\}\\
\intertext{ and the eigenvectors:}
\left(\mathbf{1},X_{2}\left(b\right),\cdots, X_{p}\left(b\right) \right)
\intertext{satisfy:}
 \forall k\in \{2,p\}:\,\, \prod_{i=1}^{p}X_k\left(b\right)_{\left(i\right)} \neq 0.
\end{align*}
\end{lemm}

\begin{proof}
Let $\mathcal{{P}}$ be a path over $p$ nodes, with all weights equal to one and $L_{\mathcal{P}}$ be its associated Laplacian matrix. 
Then denoting by $L_{p-1}$ the Laplacian matrix of the subpath over the nodes $1,\cdots,p-1$, with weights equal to one, we can write:
\begin{align*}
L_{\mathcal{P}}=
\left(\begin{array}{ccccc|c}
\\
&& &&&0\\
&&&&& \vdots \\
&&L_{p-1}+D\left(0,\cdots,0,1 \right)&&& 0\\
&&&&&\vdots \\
&&&&&-1 \\
&&&&&\\
\hline 
&&&&&\\
0 &\cdots &0&\cdots &-1&1\\
\end{array}\right),
\end{align*}
where $D\left(0,\cdots,0,1 \right)$ stands for the diagonal matrix of size $p-1$ of which entries are given by the tuple $\left(0,\cdots,0,1\right)$. Some of the components of the eigenvectors of $L_{\mathcal{P}}$ can be zero (see Remark \ref{22March}): 
the idea is to obtain the desired tuple $b$ from a perturbation of the tuple $\left(1,\cdots,1\right)$.\\
It is well known that:
\begin{align*}
\mathfrak{S}\left(L_{\mathcal{P}}\right)&=\{0< 2-2 \cos \left(\dfrac{\pi}{p}\right)<2-2 \cos \left(\dfrac{2\pi}{p}\right)<\cdots<2-2 \cos \left(\dfrac{\pi\left(p-1\right)}{p}\right)\}\\
\mathfrak{S}\left(L_{p-1}\right)&=\{0< 2-2 \cos \left(\dfrac{\pi}{p-1}\right)<2-2 \cos \left(\dfrac{2\pi}{p-1}\right)<\cdots<2-2 \cos \left(\dfrac{\pi\left(p-2\right)}{p-1}\right)\}\\
\intertext{and consequently:}
&\mathfrak{S}\left(L_{\mathcal{P}}\right) \cap \mathfrak{S}\left(L_{p-1}\right)=\{0\}.
\end{align*}
Now, if $\lambda$ is a strictly positive eigenvalue of $L_{\mathcal{P}}$, there exists a non-null vector 
$\begin{bmatrix}
X\\
z\\
\end{bmatrix}
$ such that:
\begin{align*}
&\begin{cases}
L_{p-1}X+ 
\begin{bmatrix}
\left(0\right)\\
-\lambda z
\end{bmatrix}=\lambda X\\
-\left(X_{\left(p-1\right)}-z\right)=\lambda z
\end{cases}
\intertext{and as $\lambda \notin \mathfrak{S}\left(L_{p-1}\right)$, thus $z\neq0$ and}
X&=\dfrac{1}{\text{det}\left[L_{p-1}-\lambda I_{p-1}\right]} \,^{\textbf{t}}\textbf{Com} \left[L_{p-1}-\lambda I_{p-1}\right]\begin{bmatrix}
\left(0\right)\\
\lambda z
\end{bmatrix}.
\end{align*}
We now use the same argument as in the proof of Theorem \ref{two}:
as $z\neq 0$, we have $X_{\left(p-1\right)}\neq z$ and so ${\dfrac{\partial}{\partial x}}_{\vert{x_p=1}}\lambda \left(1,\cdots,1,x_p\right) =\dfrac{\left(X_{\left(p-1\right)}-z\right)^2}{\vert \vert X\vert \vert^2+ z^2}>0$. Thus perturbing the value $x_p=1$ in $1+\epsilon_p$ we have 
that the new eigenvalue $\lambda \left(1,\cdots,1,1+\epsilon_p\right)$ (corresponding to a new eigenvector
$\begin{bmatrix}
X\left(1,\cdots,1,1+\epsilon_p\right)\\
z\left(1,\cdots,1,1+\epsilon_p\right)\\
\end{bmatrix}
$
)
is not in the spectra of the submatrices of $L_{p-1}$, and still does not belong to the spectrum of $\mathfrak{S}\left(L_{p-1}\right)$. Therefore, as in the proof of Theorem \ref{two}, the 
coordinates of this new eigenvector (associated to the eigenvalue $\lambda\left(1,\cdots,1,1+\epsilon_p\right)$) are all non-zero polynomial expressions in $\lambda\left(1,\cdots,1,1+\epsilon_p\right)$. Perturbing again the value $1+\epsilon_p$ if necessary, we get that all these polynomial expressions are distinct from zero:
\begin{align*}
 \prod_{i=1}^{p-1}X\left(1,\cdots,1,1+\epsilon_p\right)_{\left(i\right)} z\left(1,\cdots,1,1+\epsilon_p\right)\neq 0.
\end{align*}
And this holds for any eigenvector of any strictly positive eigenvalue of $L_{\mathcal{P}}\left(1,\cdots,1,1+\epsilon_p\right)$. QED.
\end{proof}

\begin{rmrk}\label{22March}
Notice that the eigenvectors of the Laplacian matrix of the undirected unweighted path with $p$ nodes are $\left(\mathbf{1},X_{2},\cdots, X_{p} \right)$, with:
\begin{align*}
{X_{k}}_{\left(i\right)}=\cos \left(\frac{\pi \left(k-1\right)i}{p}-\frac{\pi (k-1)}{2p}\right). 
\end{align*}
It can happen that an entry of $X_k$ is zero: for instance if $p=10$, $k=3$ and $i=3$. Therefore Lemma \ref{path} is not true for the tuple $\left(1,\cdots,1\right)$: we needed to perturb this tuple to prove Lemma \ref{path}.
\end{rmrk}
With Lemma \ref{path} we can now prove Lemma \ref{appendix}:

\begin{proof}[Proof of Lemma \ref{appendix}]
Let $\mathcal{T}_n$ a spanning tree of $\mathcal{G}$ and $\mathcal{P}$ the longest path inside $\mathcal{T}_n$: reordering the nodes if necessary we can assume $\mathcal{P}=\{1,2,\cdots, p\}$.
Applying Lemma \ref{path}, let us take a tuple $b=\left(b_1,\cdots,b_{p-1}\right)$ of strictly positive numbers such that:
\begin{align*}
\mathfrak{S} \left(L_{\mathcal{P}}\left(b\right) \right)= \{0<\lambda_2\left(b\right)<\cdots< \lambda_{p}\left(b\right)\}\\
\intertext{with eigenvectors:}
\left(\mathbf{1},X_{2}\left(b\right),\cdots, X_{p}\left(b\right) \right)
\intertext{satisfying:}
 \forall k\in \{2,p\}:\,\, \prod_{i=1}^{p}X_k\left(b\right)_{\left(i\right)} \neq 0.
\end{align*}
As in the proof of Theorem \ref{one} we look at the edges of the tree $\mathcal{T}_n$ that do not belong to the path $\mathcal{P}$: starting from the root $1$ of the tree, let us consider the first out-branch $\mathcal{B}_1$ of $\mathcal{T}_n$, and inside $\mathcal{B}_1$, the first of these edges, say $\mathcal{E}_1=\left(q_1,p+1\right)$, weighted by a real number $\alpha_{p}$.\\

We have:
\begin{align*}
\mathfrak{S}\left(L_{\mathcal{P}\cup \mathcal{E}_1}\left(b,\alpha_{p}\right)\right)=\{0,\lambda_2 \left(b,\alpha_{p}\right),\cdots ,\lambda_{p+1}\left(b,\alpha_{p}\right)\},
\intertext{with:} \{\lambda_2 \left( b,0\right), \lambda_3 \left( b,0\right)\cdots ,\lambda_{p+1}\left(b,0\right)\}=\{0,
\lambda_2\left( b\right),\cdots ,\lambda_{p}\left(b\right)\}.
\end{align*}
For $\alpha_{p}>0$ small enough the second eigenvalue $\lambda_2 \left(b,\alpha_{p}\right)>0$ is a simple one, since it is a very small perturbation of the double eigenvalue $0$ of $L_{\mathcal{P}\cup \{p+1\}}$. Moreover, since $\lambda_k \left(b,0\right)$ is simple we can apply again the perturbation formula to get:
\begin{align*}
\forall k \in \{3,\cdots,p+1\}:\,\, {\dfrac{\partial}{\partial x_{p}}}{\lambda_k \left(b,x_{p}\right)}_{\vert x_{p}=0}=\dfrac{X_{k-1}\left(b\right)_{(q_1)}^2}{\vert \vert X_{k-1}\left(b\right)\vert \vert ^2}>0.
\end{align*}
Therefore the spectrum of $L_{\mathcal{P}\cup \mathcal{E}_1}\left(b,\alpha_{p}\right)$ is entirely perturbed, i.e it does not intersect the spectrum of $L_{\mathcal{P}}\left(b\right)$ (except in the zero eigenvalue): 
\begin{align*}
\mathfrak{S}\left(L_{\mathcal{P}\cup \mathcal{E}_1}\left(b,\alpha_{p}\right)\right) \cap \mathfrak{S}\left(L_{\mathcal{P}}\left(b\right)\right)=\{0\}.
\end{align*}
Besides, for $3\leq k\leq p+1$, let us write the eigenvectors of $L_{\mathcal{P}\cup \mathcal{E}_1}\left(b,\alpha_{p}\right)$ as: 
\begin{align*}
\{\mathbf{1}, X_{2}\left(b,\alpha_{p}\right),\cdots, X_{p+1}\left(b,\alpha_{p}\right)\}=\{
\begin{bmatrix}
\mathbf{1}\\
1\\
\end{bmatrix},
\begin{bmatrix}
U_{2}\left(b,\alpha_{p}\right)\\
z_{2}\left(b,\alpha_{p}\right)\\
\end{bmatrix}
,\cdots, \begin{bmatrix}
U_{p+1}\left(b,\alpha_{p}\right)\\
z_{p+1}\left(b,\alpha_{p}\right)\\
\end{bmatrix}
\}.
\end{align*}
It now suffices to use the same argument as in the proof of Theorem \ref{two} and of Lemma \ref{path}. Writing the eigenvalue equation $L_{\mathcal{P}\cup \mathcal{E}_1}\left(b,\alpha_{p}\right)X_{k}\left(b,\alpha_{p}\right)=\lambda_k \left(b,\alpha_{p}\right) X_{k}\left(b,\alpha_{p}\right)$ we obtain the relation $X_{k}\left(b,\alpha_{p}\right)_{\left(q_1\right)} \neq z_{k}\left(b,\alpha_{p}\right)$ otherwise $\lambda_k \left(b,\alpha_{p}\right)$ would belong to the spectrum $\mathfrak{S}\left(L_{\mathcal{P}}\left(b\right)\right)$. Therefore $z_{k}\left(b,\alpha_{p}\right)\neq 0$ and
we have:
\begin{align*}
U_{k}\left(b,\alpha_{p}\right)&= \dfrac{1}{\text{det}\left[L_{\mathcal{P}}\left(b\right)-\lambda_k \left( b,\alpha_{p}\right)I_{p} 
\right]}{}^{\textbf{t}} \textbf{Com}\left[L_{\mathcal{P}}\left(b\right)-\lambda_k \left(b,\alpha_{p}\right)I_{p}\right] \begin{bmatrix}
\left(0\right)\\
\lambda_k \left(b,\alpha_{p}\right) z_{k}\left(b,\alpha_{p}\right)\\
\left(0\right)
\end{bmatrix},\\
\end{align*}
from which we conclude again that perturbing the value $\alpha_{p}$ if necessary, $U_{k}\left(b,\alpha_{p}\right)$ has only non- zero coordinates, and the same holds for all the eigenvectors $X_{2}\left(b,\alpha_{p}\right),\cdots,X_{p+1}\left(b,\alpha_{p}\right)$.\\

Repeating the process for the other edges of the first out-branch $\mathcal{B}_1$ and for the other out-branches
of the tree $\mathcal{T}_{n}$, we get the existence of a tuple $\mathscr{D}$ in ${\Er_+^*}^{n-2}$ for the subtree $\mathcal{T}_{n-1}$ (formed by the nodes $\left(1,\cdots,n-1\right)$) of $\mathcal{T}_{n}$
and of a real number $\alpha_{k_n}>0$ such that for the tuple $\mathscr{E}=\left(\mathscr{D},\alpha_{k_n}\right)$ we have:

\begin{align*}
&\mathfrak{S} \left(L_{\mathcal{T}_{n}}\left( \mathscr{E}\right) \right)= \{0<\lambda_2\left(\mathscr{E}\right)<\cdots< \lambda_{n}\left(\mathscr{E}\right)\}\\
& \forall k\in \{2,n\}:\,\, \prod_{i=1}^{n}X_k\left( \mathscr{E}\right)_{\left(i\right)} \neq 0,
\intertext{and:}
&\mathfrak{S} \left(L_{\mathcal{T}_{n-1}}\left( \mathscr{D}\right) \right) \cap \mathfrak{S} \left(L_{\mathcal{T}_{n}}\left( \mathscr{E}\right) \right)=\{0\}
\end{align*}
where $\{\mathbf{1},X_2\left(  \mathscr{E}\right),\cdots,X_{n} \left(\mathscr{E}\right)\}$ are the eigenvectors associated to the eigenvalues $\lambda_k\left(\mathscr{E}\right)$. 
Consequently, denoting by $Res$ the resultant map (recall that for two polynomials $P,Q$ we have $Res(P,Q)=0$ if and only if $P$ and $Q$ have a common root), we conclude the map $R_L$ defined by:
\begin{align*}
\begin{array}{ccccc}
R_L& : & \Er^{\frac{\left(n-1\right)\left(n-2\right)}{2}}\times \Er^{n-1} & \to & \Er\\
& & \left(\mathscr{E}_1,\mathscr{E}_2\right) & \mapsto & {Res}\left(\dfrac{{\chi}_{{\underline{L}}_{n-1}\left(\mathscr{E}_1\right)}}{X},\dfrac{{\chi}_{{\underline{L}}\left(\mathscr{E}_1,\mathscr{E}_2\right)}}{X}\right) \\
\end{array},
\end{align*}
is a non-null polynomial map. Indeed let us complete the tuple $\mathscr{E}$ found above by adding some $0$ so as to obtain a tuple $\left(\mathscr{E}_0,\mathscr{E}'_0\right)$ in $\Er_+^{\frac{n\left(n-1\right)}{2}}$. We have proved above that:
\begin{align*}
R_L\left(\mathscr{E}_0,\mathscr{E}'_0\right)={Res}\left(\dfrac{{\chi}_{L_{\mathcal{T}_{n-1}}\left(\mathscr{D}\right)}}{X},\dfrac{{\chi}_{L_{\mathcal{T}_n\left(\mathscr{E}\right)}}}{X}\right) \neq 0.
\end{align*}
Finally, as in the proof of Theorem \ref{one}, we thus conclude that the set of tuples $\left(\mathscr{E}_1,\mathscr{E}_2\right)$ for which we have $R_L \left(\mathscr{E}_1,\mathscr{E}_2\right)= 0$, is of Lebesgue measure $0$ in $\Er^{\frac{n\left(n-1\right)}{2}}$ and is of complement dense in $\Er^{\frac{n\left(n-1\right)}{2}}$. QED.
\end{proof}

\bibliographystyle{plain}

\end{document}